\documentclass[12pt]{iopjournal}

\usepackage{amsmath}
\usepackage{amssymb}
\usepackage{amsthm}
\usepackage{bm}
\usepackage{dcolumn}
\usepackage{mathrsfs}
\usepackage{tikz}
\usetikzlibrary{positioning, arrows.meta, calc}
\usepackage[utf8]{inputenc}
\usepackage[T1]{fontenc}
\usepackage{etoolbox}
\usepackage[numbers,sort&compress]{natbib}
\usepackage{setspace}
\DeclareMathAlphabet{\mathcal}{OMS}{cmsy}{m}{n}

\theoremstyle{plain} 
\newtheorem{theorem}{Theorem}

\theoremstyle{definition} 

\begin{document}

\articletype{Paper}

\title{Gauge Theoretic Signal Processing I: The Commutative Formalism for Single-Detector Adaptive Whitening}

\author{James Kennington$^{1,2,*}$ and Joshua Black$^{1,2}$}

\affil{$^1$Department of Physics, The Pennsylvania State University, University Park, Pennsylvania 16802, USA}

\affil{$^2$Institute for Gravitation and the Cosmos, The Pennsylvania State University, University Park, Pennsylvania 16802, USA}

\affil{$^*$Author to whom any correspondence should be addressed.}

\email{james.kennington@ligo.org}


\begin{abstract}
We present a geometric framework for adaptive whitening in gravitational-wave detectors, reformulating the problem from a sequence of spectral factorizations to parallel transport on a principal bundle.
We identify the whitening filter as a section over the manifold of power spectra and derive the \textit{minimum-phase connection} as the unique geometric structure that enforces signal causality while preserving signal-to-noise ratio.
This construction yields a rigorous definition of \textit{geometric drift}, a coordinate-independent scalar measuring the intrinsic instability of the detector noise floor.
The central result is the \textit{flatness theorem}, which proves that the curvature of the connection vanishes for scalar fields.
This establishes a \textit{holonomic update law}, guaranteeing that the optimal filter correction is path-independent and determined solely by the instantaneous noise state, free from geometric phase or hysteresis.
This framework unifies the static theory of Wiener-Hopf factorization with the dynamic requirements of real-time control, providing a rigorous certification for the stability of zero-latency calibration routines and establishing a foundation for gauge-theoretic signal processing (GTSP) in next-generation detector networks.
In the wider context of gravitational-wave astronomy, this result governs the latency and fidelity of low-latency searches and provides the single-detector foundation for the multi-detector network extensions of the framework.
The update law is integrated into a production matched-filter pipeline and validated on detector data in a companion paper.
\end{abstract}

\section{Introduction}
\label{sec:intro}

Low-latency gravitational-wave searches convert detector strain into candidate triggers under a strict time budget, so that electromagnetic observatories can slew to a sky position while a counterpart is still bright \cite{messickAnalysisFrameworkPrompt2017b, nitzPyCBCLiveRapid2018, chuSPIIROnlineCoherent2021, cannonEarlyWarningDetectionGravitational2012d}.
A leading contribution to the algorithmic latency of these searches is the whitening of the data, the transformation that flattens the detector noise floor before matched filtering \cite{tsukadaApplicationZerolatencyWhitening2018}.
Production pipelines whiten with acausal, linear-phase filters, which are numerically robust and preserve waveform shape, but which require a look-ahead data buffer whose length scales with the power spectral density (PSD) estimation window \cite{sachdevGstLALSearchAnalysis2019a, ewingPerformanceLowlatencyGstLAL2024}.
This look-ahead imposes a fixed latency of several seconds, a cost that a causal, minimum-phase whitening filter would remove, since a causal filter uses only past data and needs no buffer.

The causal, minimum-phase approach to whitening has a prior history in low-latency gravitational-wave search that motivates the present paper.
A causal, minimum-phase (``zero-latency'') whitening filter was introduced into the \textsc{GstLAL} pipeline by Tsukada et al.~\cite{tsukadaApplicationZerolatencyWhitening2018}, who demonstrated its latency advantage over the acausal, linear-phase configuration.
This filter, applied without correction for the drifting noise floor, was deployed in production for a period during the second observing run (O2).
It was subsequently withdrawn, the low-latency sky localizations it produced having been found inaccurate, a consequence of phase-recovery errors in the whitening filter as the detector noise evolved between PSD updates.
The configuration is recorded only in a methods preprint that never entered the peer-reviewed literature and was never corrected on this point~\cite{sachdevGstLALSearchAnalysis2019a}, so its withdrawal has not previously appeared in the published record.
All subsequent production analyses, through the third and fourth observing runs, have used the acausal, linear-phase architecture, and a survey of the other low-latency search pipelines, PyCBC Live, MBTA, and SPIIR, finds that none currently employs minimum-phase whitening~\cite{nitzPyCBCLiveRapid2018, alleneMBTAPipelineDetecting2025, chuSPIIROnlineCoherent2021}.
The companion paper of this series reproduces this failure mode directly and demonstrates its correction under the update rule derived here~\cite{KenningtonBlack2026_GTSP2}.

The object missing from this picture is a drift-consistent update rule for the causal filter, a prescription that adjusts the minimum-phase whitening filter as the noise floor drifts without reintroducing the phase errors that ended the O2 deployment.
This paper constructs that rule geometrically.
We identify the minimum-phase whitening filters as a section of a principal bundle of causal spectral factors, we define the minimum-phase connection as the unique parallel transport that updates the filter magnitude while accumulating zero excess phase, and we prove a flatness theorem, that the curvature of this connection vanishes for a single detector.
Flatness collapses the path-ordered transport into an exact, path-independent update law,
\begin{equation}
    \label{eq:intro_update}
    \mathcal{K} = \exp\left( \mathcal{P}_+ \left[ \log \left( \frac{S_{\mathrm{ref}}}{S_{\mathrm{live}}} \right) \right] \right),
\end{equation}
in which $\mathcal{P}_+$ is the causal Wiener-Hopf projection and $S_{\mathrm{ref}}$, $S_{\mathrm{live}}$ are the reference and live power spectra.
The correction depends only on the instantaneous noise state, not on the history of the drift, so it carries no geometric hysteresis.
The companion paper integrates this update into the production \textsc{sgnl} matched-filter pipeline and validates it with 15{,}347 binary black hole injections on O3 data across the LIGO-Virgo network, finding the detection sensitivity and sky-localization accuracy preserved and the trigger latency held constant independent of the PSD estimation window~\cite{KenningtonBlack2026_GTSP2}.
We develop the construction from the foundations of optimal detection, beginning with the matched-filter statistic and the spectral factorization that defines the whitening filter.

The detection of gravitational waves reduces fundamentally to the problem of distinguishing a deterministic waveform (signal) $h(t)$ from a stochastic background $n(t)$ \cite{helstromStatisticalTheorySignal1968}.
In the regime where physical strains are orders of magnitude smaller than the instrumental noise floor, as is the case for LIGO \cite{TheLIGOScientificCollaboration2015Advanced}, Virgo \cite{Acernese2015Advanced}, and KAGRA \cite{Akutsu2021Overview}, simple amplitude thresholding fails \cite{maggiore2007gravitational, creighton2011gravitational}.
Optimality requires a detection statistic that exploits the full spectro-temporal structure of the signal \cite{finn1992detection}.

In the domain of weak signal detection, the theory of matched filtering provides a complete and rigorous framework for maximizing the signal-to-noise ratio (SNR), $\rho$, in stationary Gaussian noise \cite{helstromStatisticalTheorySignal1968, wainsteinExtractionSignalsNoise1970, finn1992detection}.
For a signal $h(t)$ with Fourier transform $\tilde{h}(f)$, the optimal SNR is defined by the noise-weighted inner product
\footnote{
Note that throughout this work we adopt a two-sided power spectral density convention, defined over $f \in (-\infty, \infty)$, to naturally accommodate the algebraic properties of the Wiener-Hopf and loop group factorizations.
This differs by a factor of 2 from the one-sided positive-frequency PSD standard commonly used in the gravitational-wave literature.
}:
\begin{equation}
    \rho^2 = \langle h, h \rangle_S = \int_{-\infty}^\infty \frac{|\tilde{h}(f)|^2}{S(f)} df.
\end{equation}
The central operator in this framework is the whitening filter $W$, which diagonalizes the noise covariance.
If $S(f)$ is the power spectral density (PSD) of the detector noise, the whitening filter is the solution to the factorization problem \cite{kolmogorov1941interpolation, wiener1949extrapolation}:
\begin{equation}
    \label{eq:spectral_factorization}
    |W(f)|^2 = \frac{1}{S(f)}.
\end{equation}
Standard texts typically treat this as a static boundary value problem: given a fixed background $S(f)$, one computes the unique spectral factor \cite{kailathLinearEstimation2000, oppenheimDiscretetimeSignalProcessing2010, sayedFundamentalsAdaptiveFiltering2003}.

However, gravitational-wave detectors are inherently non-stationary systems \cite{abbottCharacterizationTransientNoise2016, abbottGuideLIGODetector2020, davisLIGODetectorCharacterization2021, sun2020characterization, vietsReconstructingStrainDetector2018}.
Environmental coupling, thermal drifts in the suspension systems, and scattered light transients cause the noise PSD to evolve continuously over time, $S(t, f)$ \cite{davisImprovingSensitivityAdvanced2019, zackayDetectingGravitationalWaves2021, buikemaSensitivityPerformanceAdvanced2020, covasIdentificationMitigationNarrow2018, cahillaneCalibrationUncertaintyAdvanced2017, driggersImprovingAstrophysicalParameter2019}.
Consequently, Eq.~(\ref{eq:spectral_factorization}) becomes a dynamic constraint \cite{littenbergBayesLineBayesianInference2015, biscoveanuQuantifyingEffectPower2020}, necessitating adaptive search techniques or perturbative corrections to maintain search sensitivity \cite{kenningtonPerturbativeGauge2026}.
The optimal whitening filter is no longer a static operator, but a trajectory $W(t, f)$ in the space of linear systems.

In conventional engineering practice, this non-stationarity is managed through discrete windowing and re-computation, a process that inevitably introduces computational latency and phase discontinuities at window boundaries \cite{messickAnalysisFrameworkPrompt2017b, nitzPyCBCLiveRapid2018, adamsLowlatencyAnalysisPipeline2016a, chuSPIIROnlineCoherent2021, cannonEarlyWarningDetectionGravitational2012d, sachdevGstLALSearchAnalysis2019a, alleneMBTAPipelineDetecting2025, ewingPerformanceLowlatencyGstLAL2024, mageeFirstDemonstrationEarly2021}.
Furthermore, standard attempts to linearly interpolate between these windows are geometrically unsound \cite{baggioConalDistancesRational2018, georgiouDistancesPowerSpectral2006, jiangGeodesicCurvesManifold2012}.
This is because the primary requirement of a physical whitening filter is \textit{invertibility} \cite{bocheSpectralFactorizationWhitening2005, clancey1977factorization}: the transformation must map the stochastic background to white noise without destroying physical information.
The set of invertible causal filters does not form a vector subspace, but rather a \textit{Lie group of units} \cite{rudin1991functional, gohberg1960systems}.

\begin{figure}[ht]
\centering
\includegraphics[width=0.98\textwidth]{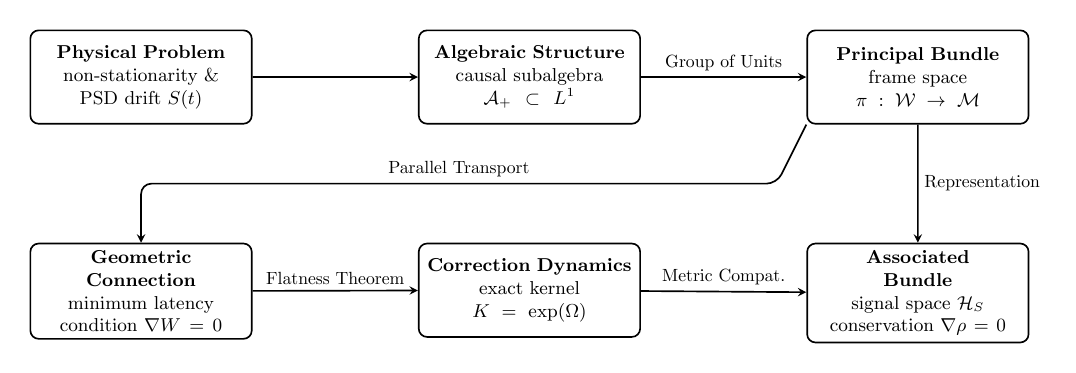}
\caption{\label{fig:logical_roadmap} Logical roadmap of the gauge-theoretic framework.
The top row establishes the geometry: the causal algebra allows us to construct a Principal Bundle of filters.
The vertical link (right) defines the Associated Hilbert Bundle of signals.
The bottom row derives the dynamics: defining parallel transport via causality leads to a flat connection, which simplifies into the Exact Correction kernel that guarantees SNR conservation ($\nabla \rho = 0$) via metric compatibility.}
\end{figure}

This work proposes a unified geometric resolution to the problem of adaptive whitening.
We identify the whitening filter as a dynamic \textit{vielbein} (a geometric frame field that diagonalizes the local signal metric), which adapts the signal space to the noise trajectory.
While the static geometry of statistical manifolds is well-understood through information geometry \cite{amariMethodsInformationGeometry2000}, we extend this to the dynamical regime by demonstrating that the problem of updating a filter is isomorphic to the problem of \textit{parallel transport} in a principal fiber bundle \cite{pancharatnam1956generalized, berry1984quantal, simon1983holonomy}.
The update rule is not an interpolation, but a covariant derivative:
\begin{equation}
    \nabla_{\dot{\gamma}} W(t) = 0.
\end{equation}
This principal connection induces a covariant derivative on the \textit{associated Hilbert bundle} of signal states \cite{kobayashi1963foundations, bleecker2013gauge}.
This geometric structure identifies the detection statistic as a conserved charge; the parallel transport of the filter guarantees the SNR, $\rho$, remains a covariant constant along the trajectory of the noise \cite{barbaresco2008innovative, amariMethodsInformationGeometry2000}:
\begin{equation}
    \nabla_{\dot{\gamma}} \rho(t) = 0.
\end{equation}
The physical gain of this approach is the unification of \textit{stability} and \textit{causality} into a single, rigorous geometric requirement.
By enforcing the horizontality of the transport, we ensure that the filter trajectory remains strictly within the manifold of causally invertible operators, providing a deep geometric explanation for the structural stability of the detector's control loop.

The physical constraint of \textit{minimum latency} (strict causality) \cite{toll1956causality} uniquely determines the specific connection form $\omega$ on the principal bundle, governing how the frame must evolve to maintain orthonormality without introducing acausal phase shifts.
The logical architecture of this framework, which links the algebraic structure of causal filters to the geometric stability of the detection statistic, is summarized in Figure \ref{fig:logical_roadmap}.

\textit{Motivation: The non-abelian horizon.}
The choice to employ the sophisticated machinery of fiber bundles for a scalar problem requires justification.
If the goal were merely to whiten a static scalar channel, classical spectral factorization would suffice \cite{wiener1949extrapolation, kailathLinearEstimation2000}.
However, the central motivation of this framework is to establish a rigorous control theory for the \textit{next generation} of detectors.
Future observatories, including terrestrial 3G facilities \cite{reitze2019cosmic, punturo2010einstein} and space-based constellations like LISA \cite{amaro2017laser}, operate as fundamentally Multi-Input Multi-Output (MIMO) sensor arrays.
In these coupled systems, the algebra becomes non-abelian and the optimal update generalizes to a path-dependent Dyson series governed by non-vanishing curvature \cite{wilczek1984appearance, simon1983holonomy}.
This work (Part I) serves as the \textit{necessary foundational limit} for that general theory.
We must rigorously prove the stability of the commutative case to define the baseline for the non-abelian gauge theory.
By solving the scalar problem geometrically, we do not just reproduce the Wiener solution; we reveal it as the zero-curvature limit of the deeper gauge-theoretic structure required for future detector networks.

\textit{The Failure of Linear Interpolation.}
One might ask why geometric transport is required at all, why not simply linearly interpolate between the reference filter $W_{\text{ref}}$ and the target $W_{\text{live}}$?
The answer lies in the structural stability of the control loop.
A naive linear update $W(t) = (1-\alpha)W_{\text{ref}} + \alpha W_{\text{live}}$ does not respect the topology of the causal subalgebra.
The sum of two minimum-phase filters is not necessarily minimum-phase; roots can migrate across the imaginary axis during the transition, introducing causal violations and instability.
By contrast, the geometric geodesic derived here guarantees that the filter trajectory remains strictly within the minimum-phase manifold $\mathcal{W}$, preserving the causal invertibility of the instrument state at every instant of the transition.

\textit{Organization of the paper.} The paper is organized as follows.
In Section \ref{sec:geometry}, we define the state spaces of the detector, comprising the manifold of admissible PSDs and the algebra of causal filters, and unify them by constructing the principal whitening bundle, identifying the unitary loop group \cite{pressley1986loop} as the gauge symmetry of spectral factorization.

Section \ref{sec:transport} addresses the dynamic control problem, deriving the minimum-phase connection as the unique geometric structure that preserves the physical matched-filter SNR without inducing dispersive group delay.
We establish the continuous parallel transport equation, demonstrating that the flow is generated by the causal projection of the logarithmic spectral drift \cite{garnettBoundedAnalyticFunctions2007, gohberg1960systems}.

Our central result, presented in Section \ref{sec:holonomic_update}, is the flatness theorem.
We prove that the curvature of the minimum-phase connection vanishes identically.
Because the base manifold is simply connected, this guarantees that parallel transport is globally path-independent, allowing us to collapse the complex path-ordered transport integral into the exact, computationally trivial holonomic update law.

Finally, Section \ref{sec:discussion} addresses the physical and operational implications of the framework.
We define the geometric drift metric as a rigorous, coordinate-independent trigger for tracking instrument stability, discuss the zero-latency deployment of the correction kernel, and outline the necessary extension of this formalism to non-abelian gauge theories for next-generation MIMO detector networks.
To maintain focus on the physical geometry, the rigorous foundations are deferred to the appendices, which detail the functional analysis of the Wiener algebra (App. \ref{app:functional_analysis}), the bundle topology and flatness proofs (App. \ref{app:bundle_topology}), the metric compatibility on the signal space (App. \ref{app:snr_conservation}), and the operational tensor representations for pipeline architecture (App. \ref{app:operational_reps}).

\textit{Structure of the Series.} This article serves as the first in a two-part series establishing the framework of gauge-theoretic signal processing (GTSP).
Here in Part I, we focus exclusively on the continuous mathematical foundation.
We treat the power spectral density as a smooth trajectory on a manifold to derive the \textit{holonomic update law}, the unique filter evolution that preserves signal causality and minimizes geometric drift.
The practical realization of this theory is presented in the companion paper, \textit{Gauge theoretic signal processing. II. Zero-latency whitening for early warning pipelines}~\cite{KenningtonBlack2026_GTSP2}.
In Part II, we address the discrete estimation of the PSD and the management of estimation lag.
We then utilize a production pipeline to validate the framework in two regimes: an offline analysis to quantify sensitivity preservation against standard linear-phase baselines and signal recovery against uncorrected minimum-latency methods; and an online analysis to verify the latency reduction in a realistic low-latency alert generation environment.

\section{The Gauge Geometry of Adaptive Whitening}
\label{sec:geometry}

To reformulate adaptive whitening as a geometric problem, we must identify the state spaces of the detector and the continuous symmetries of the spectral factorization process.
We introduce these components sequentially: the base manifold of noise states, the algebra of causal filters, and the principal bundle that unifies them.

\subsection{The Noise Manifold and Causal Algebra}
Let the configuration space of the detector's background noise be the manifold $\mathcal{M}$, defined as the space of \textit{admissible power spectral densities (PSDs)}.
An admissible PSD $S \in \mathcal{M}$ is a symmetric, strictly positive function bounded away from zero and infinity in the $L^\infty(\mathbb{R})$ norm, physically representing a noise floor that bottoms out at a constant thermal or quantum shot-noise limit.
To ensure spectral operations remain bounded, we further require $S(f)$ to be locally Hölder continuous.
The rigorous topological properties of $\mathcal{M}$, including its structure as a simply-connected convex cone and its satisfaction of the Kolmogorov-Krein criterion, are detailed in Appendix \ref{app:functional_analysis}.

For a given noise state $S \in \mathcal{M}$, the operation of whitening requires finding a filter $W$ such that the power spectrum of the filtered noise is flat: $|W(f)|^2 = S^{-1}(f)$. 

However, for $W(f)$ to be implemented in a real-time control system, its time-domain impulse response $h_W(t)$ must be strictly causal ($h_W(t) = 0$ for $t < 0$) and absolutely integrable to guarantee bounded-input bounded-output (BIBO) stability.
We therefore demand that valid physical filters reside in the \textit{causal Wiener algebra}, denoted $\mathcal{A}_+$.
For a filter to be both causal and causally invertible, it must belong to the group of units of this subalgebra, $\mathcal{G}(\mathcal{A}_+)$, which corresponds precisely to the set of strict minimum-phase filters.

\subsection{The Whitening Bundle}

\begin{figure}[t]
    \centering
    \includegraphics[width=0.98\textwidth]{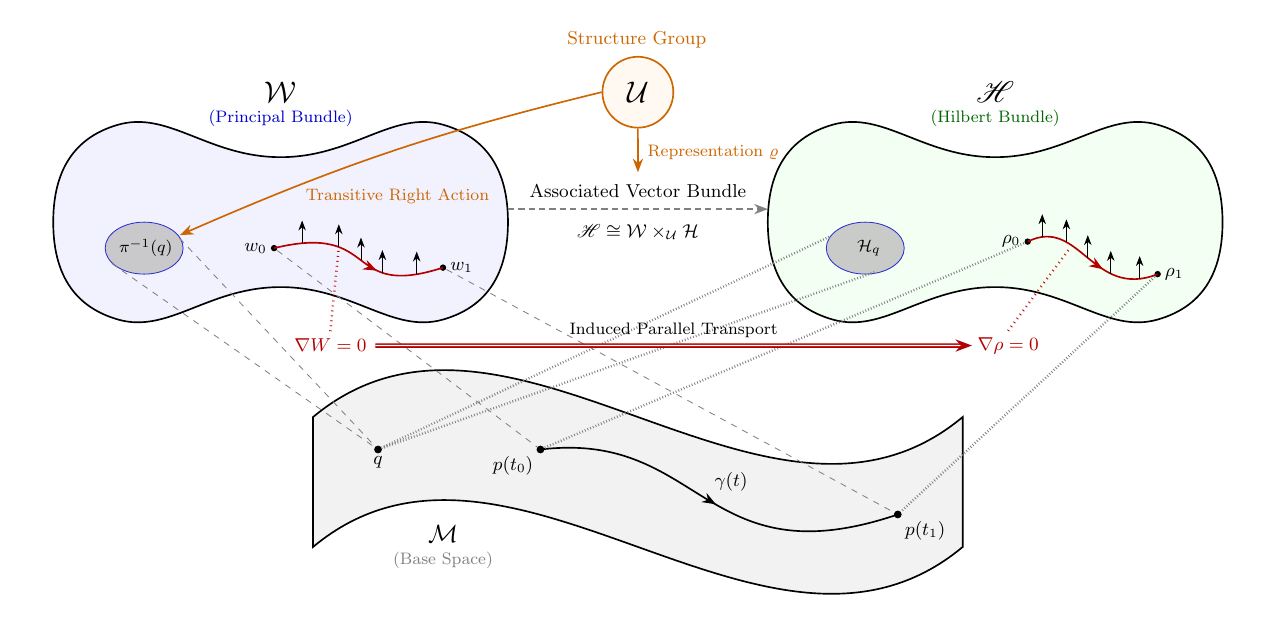}
    \caption{\label{fig:bundles}Principal bundle geometry of adaptive whitening over the noise manifold $\mathcal{M}$.
The structure group $\mathcal{U}$ coordinates the principal bundle $\mathcal{W}$ (filters) and the associated Hilbert bundle $\mathscr{H}$ (signals).
Strict causality imposes the connection defining the filter parallel transport $\nabla W = 0$; this induces the compatible signal transport $\nabla \rho = 0$, ensuring intrinsic SNR conservation along the trajectory $\gamma(t)$.}
\end{figure}

The spectral magnitude condition $|W|^2 = S^{-1}$ does not uniquely determine the whitening filter.
If $W$ is a valid frame, any gauge transformation $W \mapsto W \cdot U$ preserves the whitening condition, provided $U(f)$ is a pure phase rotation, $|U(f)| = 1$. 

We parameterize this continuous symmetry using the structure group $\mathcal{U}$.
By applying a one-point compactification to the frequency domain ($f \to \pm\infty$), we identify $\mathcal{U}$ as a dense, stationary Lie subgroup of the infinite-dimensional unitary loop group, $L\mathrm{U}(1)$.
The precise boundary conditions that prevent $1/t$ instability, alongside the group's Birkhoff factorization properties, are explicitly constructed in Appendix \ref{app:functional_analysis}.

This local phase invariance naturally constructs a principal fiber bundle over the noise manifold, which we term the \textbf{whitening bundle} (illustrated in Figure \ref{fig:bundles}):
\begin{equation}
    \mathcal{W} \xrightarrow{\pi} \mathcal{M}.
\end{equation}
The base space $\mathcal{M}$ represents the physical magnitude of the noise, while the vertical fiber $\mathcal{W}_S \cong \mathcal{U}$ over any given $S \in \mathcal{M}$ contains the set of exact whitening filters that differ only by their phase response.

To operate the detector, the control system must continuously select one specific filter from each fiber.
Geometrically, this choice constitutes a \textit{global section} of the bundle, $\sigma: \mathcal{M} \to \mathcal{W}$.
Because latency is the critical metric in gravitational-wave transient alerts, we define the \textbf{minimum-phase section}, $\sigma_{mp}$.
By applying the causal Riesz projection $\mathcal{P}_+$ to the logarithm of the PSD, we isolate the mathematically unique filter in the fiber that perfectly whitens the noise while accumulating zero excess geometric phase.
In standard signal processing, the additive algebra of these logarithmic spectral projections provides the theoretical foundation for the complex cepstrum.
Geometrically, this section serves as a global trivialization of the whitening bundle, ensuring that every noise state is assigned a strictly causal, minimum-latency reference frame.

\section{The minimum-phase connection and parallel transport}
\label{sec:transport}

Having established the whitening bundle, we now address the dynamic problem of adaptive control.
As the physical environment of the detector evolves, the true power spectrum traces a continuous trajectory $\gamma(t)$ through the base manifold $\mathcal{M}$.
To maintain optimal sensitivity, the control system must continuously update the whitening filter $W(t)$ along the fibers above this trajectory.
Geometrically, this requires equipping the bundle with an Ehresmann connection, which provides a rule for parallel transport \cite{atiyah1978geometry, kobayashi1963foundations, bleecker2013gauge}.

\subsection{The minimum-latency constraint}
A connection smoothly separates the tangent space of the bundle into a vertical subspace (pure phase gauge transformations) and a horizontal subspace (physical changes in noise magnitude).
Transporting a filter "horizontally" means updating its magnitude to match the new noise state without accumulating arbitrary gauge phase.

Because the structure group $\mathcal{U}$ is unitary, \textit{any} connection will perfectly preserve the scalar whitening condition $|W|^2 = S^{-1}$.
However, from a physical signal-processing perspective, an arbitrary connection is catastrophic.
If the parallel transport dictates a gauge transformation that accumulates non-linear phase, it induces frequency-dependent group delay.
This dispersion smears the time-domain wave packet of the gravitational-wave template, destroying the peak of the matched-filter cross-correlation.
Even a linear phase accumulation, while preserving the wave shape, introduces a spurious time-of-arrival shift that degrades the astrophysical localization of the source network.

To guarantee that the parallel transport preserves both the exact waveform shape and its true astrophysical arrival time, we demand that the transport accumulates absolutely zero excess geometric phase.
We formalize this by defining the \textbf{minimum-phase connection}.
We declare the global minimum-phase section $\sigma_{mp}$ defined in Section \ref{sec:geometry} to be horizontal.
Therefore, a filter is parallel transported along a noise trajectory if and only if its covariant derivative vanishes: $\nabla_{\dot{\gamma}} W = 0$.

While other Ehresmann connections are mathematically possible, defined by horizontal distributions associated with different subgroups of the unitary loop group, they are physically prohibited by the requirements of low-latency detection.
For example, a \textit{linear-phase connection} would accumulate a frequency-dependent phase drift during transport; while this would preserve the spectral magnitude, it would introduce an artificial, time-varying time-of-arrival shift that would decouple the detector from the global timing grid.
More general dispersive connections would induce non-linear group delay, smearing the signal power and reducing the effective search sensitivity.
The minimum-phase connection is thus the unique geometric structure that unifies the requirements of spectral optimality, stability, and zero-latency timing into a single covariant requirement.

\subsection{The transport equation}
To determine how the filter must dynamically evolve to satisfy this horizontal constraint, we evaluate the connection 1-form.
As derived rigorously via the Maurer-Cartan structure equations in Appendix \ref{app:bundle_topology}, the instantaneous algebraic generator of the transport, $\Omega(t) = \dot{W}W^{-1}$, evaluates exactly to:
\begin{equation}
    \Omega(t) = \mathcal{P}_+\left[- \frac{\partial}{\partial t} \log S(t) \right].
\end{equation}
This establishes a fundamental physical result: the minimum-latency transport of a filter is generated purely by the causal projection of the instantaneous logarithmic drift of the power spectrum.

To update a reference filter $W(0)$ to the live state $W(t)$ over a finite time interval, we integrate this generator along the trajectory $\gamma$.
The solution to this differential equation is the parallel transport operator, or \textbf{correction kernel} $\mathcal{K}_\gamma$.
Formally, this is given by the path-ordered exponential (the Dyson series):
\begin{equation}
    \label{eq:dyson}
    \mathcal{K}_\gamma = \mathcal{T} \exp \left( \int_{\gamma} \Omega(t) \, dt \right),
\end{equation}
where $\mathcal{T}$ denotes time-ordering along the path.
As we shall demonstrate in Section \ref{sec:holonomic_update}, the curvature of this connection vanishes identically for scalar fields.
This flatness allows the complex path-ordered integral to collapse into a trivial, path-independent state-to-state update law, which serves as a central verifiable prediction of the GTSP framework.

\subsection{Intrinsic SNR conservation}
The transport of the filter natively induces a compatible transport on the physical signals themselves.
A whitening filter acts as a local frame (a \textit{vielbein}), mapping the noise-weighted Hilbert space of the curved base manifold to a flat Euclidean space.
As proven in Appendix \ref{app:snr_conservation}, the condition $\nabla_{\dot{\gamma}} W = 0$ strictly guarantees that the covariant derivative of the optimal SNR vanishes: $\nabla_{\dot{\gamma}} \rho = 0$. 

Thus, continuous parallel transport via the minimum-phase connection provides an exact geometric guarantee: the pipeline intrinsically conserves the theoretical optimum matched-filter SNR of the evolving detector without incurring any dispersive penalties.

\section{The flatness theorem and holonomic update Law}
\label{sec:holonomic_update}
While Section III established that the parallel transport operator $\mathcal{K}_\gamma$ (the Dyson series in Eq. \ref{eq:dyson}) successfully conserves the exact SNR, evaluating a path-ordered integral in a real-time control system is generally intractable.
It implies that the necessary filter correction depends on the exact historical trajectory of the noise floor, meaning the system exhibits geometric hysteresis.
To determine whether the filter actually retains memory of how the noise drifted, we must evaluate the holonomy of the connection.

\subsection{The flatness theorem}
The holonomy of a principal bundle is governed by the curvature 2-form $\Theta$ of its connection.
If the curvature is non-zero, parallel transport is path-dependent, and the optimal filter requires continuous, integrated tracking of the noise evolution.
If the curvature vanishes identically, the connection is flat, and transport between any two states is entirely path-independent. 

We formalize the structural stability of the single-detector whitening problem in the following theorem:

\begin{theorem}[The Flatness Theorem]
\label{thm:flat}
The minimum-phase connection $\omega_{mp}$ on the whitening bundle $\mathcal{W}$ has vanishing curvature, $\Theta \equiv 0$.
Consequently, the parallel transport operator $\mathcal{K}_\gamma$ is invariant under path homotopy.
Because the base manifold of admissible power spectra $\mathcal{M}$ is a simply connected convex cone ($\pi_1(\mathcal{M}) = 0$), the transport is strictly and globally path-independent.
\end{theorem}

The rigorous algebraic and geometric proofs of this theorem are provided in Appendix \ref{app:bundle_topology}.
Physically, this flatness arises from two structural properties: first, the structure group of scalar filters is abelian (commutative), meaning the Lie bracket of the algebra vanishes ($[\omega_{mp}, \omega_{mp}] = 0$); second, the causal horizontal distribution defined by the minimum-phase section is naturally involutive, guaranteeing that the exact differential term vanishes ($d\omega_{mp} = 0$).

\subsection{The holonomic update law}
The global path-independence guaranteed by the flatness theorem yields a measurable operational advantage.
Because the transport operator $\mathcal{K}_\gamma$ is entirely independent of the specific trajectory $\gamma(t)$ taken by the noise floor, the path-ordered Dyson series collapses into a simple, exact state-to-state evaluation.

Integrating the causal generator $\Omega$ directly between a static reference state ($S_{ref}$) and the instantaneous live state ($S_{live}$) yields the exact \textbf{holonomic update law}:
\begin{equation}
    \label{eq:holonomic}
    \mathcal{K} = \exp\left( \mathcal{P}_+ \left[ \log \left( \frac{S_{ref}}{S_{live}} \right) \right] \right).
\end{equation}

For the single-detector, or single-input-single-output (SISO), case, Equation~(\ref{eq:holonomic}) represents the foundational engineering result of the commutative formalism.
It demonstrates that the exact optimal correction required to recover the theoretical maximum SNR depends \textit{only} on the ratio of the instantaneous live noise to the reference noise.
The control system requires no memory of the detector's prior states and no iterative gradient descent; the pipeline simply projects the logarithmic ratio into the causal Wiener algebra and exponentiates the result to bridge the gap between the static templates and the live detector.

It is crucial to emphasize that this functional simplicity, the collapse of the path-ordered integral into a trivial state-to-state ratio, is specifically a consequence of the abelian nature of a single data stream.In MIMO environments, such as fully coherent global detector networks or arrays with correlated noise, the structure group becomes non-Abelian (e.g., $L\mathrm{U}(N)$).
In these higher-dimensional spaces, the algebra is non-commutative ($[\omega, \omega] \neq 0$), the connection is no longer flat, and the optimal filter \textit{does} accumulate geometric hysteresis .
The rigorous principal bundle machinery developed here is therefore not merely a relabeling of standard signal processing, but the necessary mathematical foundation required to systematically track and correct path-dependent geometric phase in those complex, non-commutative limits.

\section{Physical Implications and Future Directions}
\label{sec:discussion}

The reformulation of adaptive whitening as parallel transport on a principal bundle is not merely a mathematical curiosity; it directly resolves the operational tension between non-stationary detector environments and the strict low-latency requirements of multi-messenger astrophysics.
By identifying the exact holonomic update law, the GTSP framework provides several immediate physical and architectural benefits, while also establishing the necessary foundation for tracking quantum-correlated detector networks.

\subsection{The geometric drift metric}
To operate the exact update law efficiently, a control system must know \textit{when} the live detector has deviated far enough from the reference state to warrant a correction.
Standard Euclidean distances between power spectra fail to capture the causal phase dynamics of the filters.
However, the geometric framework naturally provides a coordinate-independent scalar measuring the intrinsic instability of the noise floor: the \textbf{geometric drift metric}, $\mathcal{D}(t)$.

We equip the statistical manifold $\mathcal{M}$ with the fisher information metric, which assigns a distance to spectral fluctuations weighted by the inverse square of the noise power \cite{amariMethodsInformationGeometry2000, georgiouDistancesPowerSpectral2006}:
\begin{equation}
    g_S(\delta S_1, \delta S_2) = \frac{1}{2} \int_{-\infty}^{\infty} \left(\frac{\delta S_1(f)}{S(f)}\right) \left(\frac{\delta S_2(f)}{S(f)}\right) df.
\end{equation}
Functionally, this metric identifies the logarithmic derivative $\dot{S}/S$ as the fundamental geometric velocity.
By defining the instantaneous drift rate as the operator norm of the causal transport generator $\Omega(t)$, we can directly express the geometric drift magnitude in terms of this fisher velocity:
\begin{equation}
    \label{eq:drift_metric}
    \mathcal{D}(t) \equiv \| \Omega(t) \| = \sqrt{ g_S(\dot{S}, \dot{S}) }.
\end{equation}
Physically, $\mathcal{D}(t)$ measures the rate at which the optimal whitening frame is rotating within the vertical fiber of the bundle.
It provides pipeline operators with a rigorous, scalar trigger: rather than updating filters on an arbitrary fixed schedule, the holonomic correction $\mathcal{K}$ can be dynamically applied only when the integrated drift $\int \mathcal{D}(t) dt$ exceeds a predefined tolerance bound.
This ensures optimal computational resource allocation while strictly bounding the maximum allowable SNR loss.

\subsection{Operational causality and latency}
Once the geometric drift exceeds the allowable tolerance, the control system must deploy the correction kernel $\mathcal{K}$ without disrupting the ongoing search.
In the era of automated transient alerts, any signal processing operation that delays the time-of-arrival estimation degrades the pointing accuracy of the global detector network \cite{fairhurstTriangulationGravitationalWave2009}.
As detailed in Appendix D, the exact geometric correction $\mathcal{K}$ can be deployed into real-time pipelines via mathematically isomorphic forward, adjoint, or post-processing representations.

While Appendix \ref{app:operational_reps} concludes that applying the strictly anti-causal adjoint filter directly to the data stream is computationally optimal for massive template banks, doing so inherently incurs a look-ahead latency penalty.
However, because the minimum-phase connection strictly forbids the accumulation of arbitrary gauge phase, the correction kernel modifies the amplitude of the matched filter without shifting its peak in the time domain.
This guarantees that whether the pipeline applies the active pushforward to the template bank, or the highly scalable exterior convolution to the output SNR time series, the geometric update incurs absolutely zero dispersive latency penalty.
The true astrophysical arrival time of the gravitational wave is strictly anchored, even as the detector's sensitivity dynamically evolves.

Furthermore, the production of the calibrated $h(t)$ stream inherently involves acausal inverse-sensing filters. 
Because the calibration pipeline already manages look-ahead buffers, an optimal, acausal minimum-phase whitening filter could theoretically be integrated directly into this calibration stage \cite{vietsReconstructingStrainDetector2018}. 
This would produce a natively whitened data stream for downstream searches without introducing any additional latency penalty, potentially incorporating perturbative corrections for instantaneous spectral drift \cite{kenningtonPerturbativeGauge2026}.

\subsection{Beyond the flat limit: non-Abelian GTSP}
While the present work focuses on the SISO case, framing this problem in the language of gauge theory is critical for the future of gravitational-wave astronomy.
The functional simplicity of the holonomic update law, defined by the collapse of the path-ordered Dyson series into a trivial state-to-state ratio, is specifically a consequence of the Abelian (commutative) nature of a single data stream. 

As the global network upgrades to next-generation facilities such as Cosmic Explorer \cite{reitze2019cosmic}, the Einstein Telescope \cite{punturo2010einstein}, and LISA \cite{amaro2017laser}, the configuration space of the background noise will become significantly more complex.
For multi-channel vector $h(t)$ streams, the injection of frequency-dependent squeezed vacuum states introduces highly sensitive cross-channel quantum noise couplings \cite{caves1981quantum, mccullerSqueezedLightAdvanced2020}.
Further, the coherent cross-correlation of multi-link networks, such as the Einstein Telescope's triangular configuration or LISA's time-delay interferometry, and the subtraction of massive environmental sensor arrays will fundamentally transition the data streams into a coupled MIMO regime.

In these regimes, the scalar structure group $L\mathrm{U}(1)$ expands to the non-Abelian unitary group $L\mathrm{U}(N)$.
Consequently, the Lie algebra of the filters ceases to commute ($[\omega, \omega] \neq 0$).
This fundamental breakdown of commutativity introduces the potential for non-vanishing curvature ($\Theta \neq 0$) and path-dependent geometric hysteresis.
Whether a global, flat minimum-phase section can be uniquely defined for matrix-valued power spectral densities, which would allow the exact differential term $d\omega$ to perfectly cancel the commutator, remains area of active research.
If such a section cannot be globally constructed, parallel transport becomes strictly path-dependent, and the optimal MIMO filter will depend fundamentally on the specific historical trajectory of the environmental noise.

Conventional optimal filtering cannot track this history without continuous, computationally prohibitive recalculations of the multidimensional Wiener-Hopf matrix equations.
However, the gauge-theoretic signal processing framework anticipates this complexity.
By elevating the problem to a principal bundle, GTSP provides the exact differential geometry required to evaluate the connection's curvature and manage any path-ordered hysteresis, securing optimal sensitivity for the next generation of observatories.
Further, in these non-Abelian regimes, the scalar geometric drift metric $\mathcal{D}(t)$ naturally generalizes to a matrix invariant of the multi-channel generator, while any non-vanishing field strength $\Theta$ will emerge as the rigorous physical metric for the accumulated path-dependence of the detector array.

Beyond gravitational-wave instrumentation, the holonomic update applies to any streaming setting that requires causal whitening against a spectrum estimated online, including real-time spectral estimation in other sensing domains.
The non-Abelian generalization to coupled sensor networks outlined above is the subject of forthcoming work.

\section{Conclusion}
By treating the power spectrum as a continuous physical manifold and the whitening filters as a principal fiber bundle, we have derived the exact, minimum-latency geometry of single-detector adaptive whitening.
The vanishing curvature of the minimum-phase connection yields a globally path-independent holonomy, allowing search pipelines to dynamically bridge the gap between static reference templates and live, non-stationary data streams without hysteresis or dispersive latency.
Having established the fundamental geometry of adaptive whitening, Part II of this work translates these continuous holonomic updates into the operational architecture required for zero-latency matched filtering.

\appendix
\section{Functional Analytic Foundations}
\label{app:functional_analysis}

This appendix provides the rigorous functional analytic definitions and topological proofs underlying the noise manifold and the causal loop algebra utilized in the main text.

\subsection{Topology of Admissible PSDs}
\label{subsection:top_admissible_psd}

The mathematical structure of the noise manifold strictly inherits its properties from the physical dynamics of the continuous-time noise process, $n(t)$.
We model $n(t)$ as a real-valued, zero-mean, wide-sense stationary (WSS) stochastic process.
By the Wiener-Khinchin theorem, the frequency-domain representation of this noise is governed by its power spectral density (PSD), $S(f)$, defined as the Fourier transform of the time-domain autocorrelation function. 

The physical nature of the detector imposes three inescapable constraints on $S(f)$.
First, because $n(t)$ is purely real, its autocorrelation is an even function, strictly requiring $S(f)$ to be real-valued and symmetric (even).
Second, the detector operates subject to standard quantum limits, possessing an irreducible, quantum noise floor that prevents the total spectral power from vanishing at any frequency.
Finally, the macroscopic noise sources (e.g., thermal fluctuations, seismic vibrations, and suspension resonances) exhibit decaying time correlations governed by physical inertia.
This guarantees not only that their colored contribution resides in $L^1(\mathbb{R})$, but that the total spectral power is globally twice continuously differentiable ($C^2$) and its colored component naturally rolls off at high frequencies with a strict asymptotic inertial decay of order 2.

Consequently, we formally define the physical subalgebra $\mathcal{W}_{\mathrm{phys}} \subset \mathbb{C} \oplus W(\mathbb{R})$ (the unitized Wiener algebra) as the space of symmetric functions $S(f) = c + w(f)$ satisfying this global $C^2$ smoothness and order-2 decay.
Mathematically, this requires the limits 
\begin{equation}
    \lim_{f \to \pm\infty} f^{2+k} w^{(k)}(f) = L_k^{\pm}
\end{equation}
to exist and be finite for $k \in \{0, 1, 2\}$.
Because $S(f)$ is an even function, $w'(f)$ is odd and $w''(f)$ is even.
Consequently, the asymptotic constants strictly match at both extremes ($L_k^+ = L_k^-$), which we denote simply as $L_k$. 

The configuration space of the background noise, $\mathcal{M}$, is formally identified as the strictly positive convex subset within this subalgebra, requiring $S(f) \ge c > 0$ for all $f \in \mathbb{R}$ to account for the broadband quantum floor. 

Because $w = \widehat{R}$ with $R \in L^1(\mathbb{R})$, the Riemann--Lebesgue lemma implies $w \in C_0(\mathbb{R})$, making it bounded and uniformly continuous.
Since $S(f) \ge c > 0$, it follows that $\log S(f)$ is globally bounded.
Moreover, because the unitized space $\mathbb{C} \oplus W(\mathbb{R})$ is a commutative Banach algebra under pointwise multiplication, the Wiener--Lévy theorem guarantees that $\log S$ also safely resides within $\mathbb{C} \oplus W(\mathbb{R})$. 

Topologically, $\mathcal{M}$ possesses three critical properties governing the stability of the control system.
First, $\mathcal{M}$ forms an open subset in the $\mathbb{C} \oplus W(\mathbb{R})$ topology; infinitesimally perturbed noise states remain physically admissible, invertible, and bounded-input bounded-output (BIBO) stable.
Second, it is dense in the space of symmetric, continuous, positive functions satisfying $S-c \in C_0(\mathbb{R})$, allowing rational approximations to converge uniformly to valid detector spectra.
Finally, because $\mathcal{M}$ is a strictly convex set, it is globally contractible and trivially simply connected ($\pi_1(\mathcal{M}) = 0$).

While $\mathcal{M}$ is natively defined on the infinite real line, the gauge-theoretic formulation of spectral factorization requires a compact base space to leverage the rigorous geometry of principal bundles and loop groups.
This physical setup provides the exact analytical constraints necessary to construct a mathematically safe compactification.

Regarding physical detector constraints, strain data is strictly band-limited to an operational band $(f_{\mathrm{low}}, f_{\mathrm{high}})$ by high-pass and anti-aliasing filters, avoiding asymptotic divergences at $f \to 0$, such as Newtonian and radiation pressure noise, and $f \to \infty$, such as cavity pole scaling. 
Setting the out-of-band PSD strictly to zero, however, violates the Paley-Wiener condition, topologically obstructing causal minimum-phase factorization. 
To preserve the geometric formalism, we assume a standard regularization where out-of-band data is smoothly tapered to a strictly positive constant noise floor $S(f) = c > 0$ (such as the digital quantization limit). 
This bounds the physical PSD as $0 < \epsilon \le S(f) \le M < \infty$, simply embedding the regularized PSD into the manifold $\mathcal{M}$.

\subsection{The Compactification Isomorphism}
\label{subsection:compactification}

To construct a valid principal bundle, the phase ambiguity of spectral factorization must be parameterized by a structure group.
While this is naturally handled on closed manifolds using loop groups, the Beurling-Helson theorem precludes a universal Banach algebra isomorphism between the continuous Wiener algebra $W(\mathbb{R})$ and the discrete Wiener algebra $W(S^1)$ under nonlinear coordinate changes.
However, because our physical system operates strictly within the smooth, inertially decaying subalgebra $\mathcal{W}_{\mathrm{phys}}$ established above, we can construct an exact algebraic embedding from the physical noise manifold on the real line into a smooth loop algebra on the circle.

We define the one-point compactification via the Cayley transform, $\theta = 2\arctan(f)$, mapping the extended real frequency line to the circle $S^1 \cong [-\pi, \pi] / \sim$, sending the high-frequency limits $f \to \pm\infty$ to the pole at $\theta = \pm\pi$.

\vspace{0.5em}
\noindent \textbf{Theorem 1 (The Subalgebra Isomorphism).}
\textit{The pullback of the Cayley transform induces an exact algebraic isomorphism and topological homeomorphism (with respect to the $C^2$ norm) between the physical line algebra $\mathcal{W}_{\mathrm{phys}}$ and the stationary smooth circle algebra, defined as $C^2_*(S^1) = \{ \tilde{S} \in C^2(S^1) \mid \tilde{S}'(\pm\pi) = 0 \}$.
Consequently, this mapping embeds the physical noise manifold strictly into the discrete Wiener algebra $W(S^1)$.}
\vspace{0.5em}

\begin{proof}
Let $S \in \mathcal{W}_{\mathrm{phys}}$.
The forward mapping is $\tilde{S}(\theta) = S(\tan(\theta/2))$.
Because the Cayley transform is a smooth diffeomorphism on $(-\pi, \pi)$, $\tilde{S}$ is strictly $C^2$ on $S^1 \setminus \{\pm\pi\}$.
Applying the chain rule, the first derivative is $\tilde{S}'(\theta) = S'(f) \frac{1+f^2}{2}$.
Substituting the asymptotic decay $S'(f) = L_1 f^{-3} + o(f^{-3})$ yields $\lim_{|f| \to \infty} \tilde{S}'(\theta) = \lim_{|f| \to \infty} \frac{L_1}{2} f^{-1} = 0$.
Thus, $\tilde{S}'(\pm\pi) = 0$.

The second derivative evaluates to:
\begin{equation}
    \tilde{S}''(\theta) = S''(f) \frac{(1+f^2)^2}{4} + S'(f) \frac{f(1+f^2)}{2}.
\end{equation}
Substituting the strict limits $S''(f) \sim L_2 f^{-4}$ and $S'(f) \sim L_1 f^{-3}$ yields $\lim_{f \to \pm\infty} \tilde{S}''(\theta) = \frac{L_2}{4} + \frac{L_1}{2}$.
Because physical symmetry enforces $L_k^+ = L_k^-$, the limits from the left ($\theta \to \pi^-$) and right ($\theta \to -\pi^+$) evaluate to the identical finite constant.
Therefore, the second derivative is continuous across the pole, establishing $\tilde{S} \in C^2_*(S^1)$ globally.

Conversely, for surjectivity, let $\tilde{S} \in C^2_*(S^1)$.
The inverse mapping is $S(f) = \tilde{S}(2\arctan(f))$.
By Taylor's theorem at the pole, utilizing the stationary condition $\tilde{S}'(\pm\pi) = 0$, we have $\tilde{S}(\theta) = \tilde{S}(\pm\pi) + \frac{1}{2}\tilde{S}''(\pm\pi)(\theta \mp \pi)^2 + o(|\theta \mp \pi|^2)$.
Substituting the asymptotic relation $\theta \mp \pi = -2/f + \mathcal{O}(f^{-3})$ directly yields $S(f) - c = 2\tilde{S}''(\pm\pi)f^{-2} + o(f^{-2})$, satisfying the $L_0$ constraint.
Applying the chain rule symmetrically to the inverse map confirms the $L_1$ and $L_2$ decay conditions, proving $S \in \mathcal{W}_{\mathrm{phys}}$.

Because $\tilde{S} \in C^2(S^1)$, its second derivative $\tilde{S}''$ resides in $L^2(S^1)$.
By Parseval's theorem, the Fourier coefficients of the second derivative, $c''_n = -n^2 c_n$, are square-summable.
Applying the Cauchy-Schwarz inequality yields $\sum_{n \neq 0} |c_n| \le (\sum_{n \neq 0} n^{-4})^{1/2} (\sum_{n \neq 0} |c''_n|^2)^{1/2} < \infty$, guaranteeing absolute summability ($\{c_n\} \in \ell^1(\mathbb{Z})$).
This guarantees $C^2_*(S^1) \subset W(S^1)$.
Furthermore, because the mapping and its inverse are defined by smooth coordinate transformations with matching boundary derivatives, the bijection is a homeomorphism in the $C^2$ topology.
Since the $C^2$ norm strictly dominates the Wiener norm ($\| \cdot \|_W \le K \| \cdot \|_{C^2}$), this topological equivalence guarantees that smooth geometric operations on the circle inherently preserve absolute stability on the real line.
Since the bijection preserves pointwise multiplication, it acts as an exact algebraic isomorphism between the specified subalgebras.
\end{proof}
\vspace{0.5em}

\noindent \textbf{Remarks on Geometry and Signal Processing.} \\
This isomorphism serves as the foundational bridge of our geometric framework, unifying two distinct paradigms:
\begin{enumerate}
    \item \textbf{Digital Signal Processing:} The Cayley transform evaluated on the imaginary axis represents the classical Bilinear Transform.
    Theorem 1 rigorously formalizes the engineering intuition that strictly proper continuous-time filters map uniquely to discrete-time filters with a stationary derivative at the Nyquist frequency ($\tilde{S}'(\pm\pi) = 0$).
    \item \textbf{Gauge Theoretic Stability:} By establishing a topological homeomorphism in the dominating $C^2$ norm, Theorem 1 guarantees that we may rigorously translate physically admissible noise states to the closed circle, perform our geometric operations within the robust topological framework of loop groups and discrete Hardy spaces, and safely pull smooth geometric constructs—such as connection forms and parallel transport trajectories—back to the physical time domain without ever violating bounded-input bounded-output (BIBO) stability.
\end{enumerate}

\subsection{The Causal Algebra and the Structure Group}
To guarantee bounded-input bounded-output (BIBO) stability, valid physical whitening filters must natively reside in the continuous time-domain Wiener algebra, $\mathcal{A}$ (the Fourier-isomorphic image of $\mathbb{C} \oplus W(\mathbb{R})$) \cite{wiener1933fourier, gohberg1960systems}.
This Banach algebra consists of distribution-valued functions of the form $h_W(t) = c\delta(t) + w(t)$, where $h_W(t)$ is the impulse response of $W(f)$, $c \in \mathbb{C}$ is the instantaneous gain and $w(t) \in L^1(\mathbb{R})$ is the absolutely integrable memory kernel.
The causal subalgebra, $\mathcal{A}_+$, restricts these functions to those strictly supported on the non-negative real line ($t \geq 0$).
Via the Laplace transform, $\mathcal{A}_+$ is rigorously identified as a subalgebra of the Hardy space $H^\infty(\mathbb{C}^+)$ of bounded analytic functions in the complex upper half-plane \cite{hoffmanBanachSpacesAnalytic1962, garnettBoundedAnalyticFunctions2007}.
For a filter to be both causal and causally invertible, it must belong to the group of units of this subalgebra, denoted $\mathcal{G}(\mathcal{A}_+)$, which corresponds precisely to the set of strict minimum-phase filters.

The phase ambiguity inherent to spectral factorization is parameterized by the structure group.
By the Beurling-Lax theorem \citep{helson2014lectures, beurling1949linear}, the shift-invariant subspaces characterizing causal phase ambiguities are generated by \textit{inner functions} (all-pass filters).
However, the semigroup of inner functions lacks inverses within $\mathcal{A}_+$, which prohibits the bidirectional gauge transformations required to define a valid principal bundle. 

To overcome this, we leverage Theorem 1 to map the filter symmetries onto the compact topology of the circle.
We define our gauge symmetry as an infinite-dimensional unitary loop group natively on $S^1$.
To ensure the group remains in the connected component of the identity, permitting a globally defined logarithm, we restrict the loops to those with zero winding number.
Accordingly, we define the structure group as:
\begin{equation}
    \mathcal{U} = \left\{ U: S^1 \to \mathrm{U}(1) \;\middle|\; U \in C^2_*(S^1), \ \mathrm{wind}(U) = 0 \right\}.
\end{equation}
Notice that the strict $C^2$ homeomorphism established in Theorem 1 is the exact algebraic mechanism that guarantees smooth loops pull back safely to the physical time-domain Wiener algebra, rigorously preserving bounded-input bounded-output stability.

\subsection{Birkhoff Factorization and the Intertwined Riesz Projection}
As an abelian Banach loop group equipped with the $C^2_*(S^1)$ topology, $\mathcal{U}$ admits a global Birkhoff factorization \citep{pressley1986loop}.
This powerful geometric theorem guarantees that any loop mapped from our physical noise manifold can be uniquely decomposed into strictly causal, anti-causal, and zero-mode components. 

Because the loop group is Abelian, the logarithmic isomorphism provided by the exponential map $\exp: L\mathfrak{u}(1) \to \mathcal{U}$ allows us to perform this geometric decomposition linearly.
Because the Cayley transform acts as a strict conformal equivalence between the unit disk $\mathbb{D}$ and the complex upper half-plane $\mathbb{C}^+$, the projection onto the causal Hardy space intertwines directly with the continuous analytic signal operator utilized in standard time-domain signal processing. 
In the continuous domain, the true idempotent causal projection onto the Hardy space acts strictly as 
\begin{equation}
\mathcal{P}_+ = \frac{1}{2}(I + i\mathcal{H})
\label{eq:riesz-projection}
\end{equation}
where $\mathcal{H}$, is the Hilbert transform. 

The topological equivalence established in Theorem 1 guarantees that applying this projection and pulling back to the real line safely preserves absolute stability and strict causality ($t \ge 0$) in the time-domain Wiener algebra.

Because minimum-phase functions are strictly analytic and non-vanishing in $\mathbb{C}^+$ (and symmetrically in $\mathbb{D}$), the argument principle guarantees their winding number around the origin is identically zero.
This trivial topology ensures a globally defined, single-valued complex logarithm, rigorously eliminating the $2\pi i n$ branch indeterminacy that frequently plagues general spectral estimation.

Thus, the operation $W_{mp} = \exp(\mathcal{P}_+[-\log S])$ globally and smoothly maps any noise state $S \in \mathcal{M}$ to its causal spectral factor.
To establish uniqueness geometrically, suppose another causal filter $W$ satisfies $|W|^2=S^{-1}$.
Taking logarithms, the difference between the two logarithmic factorizations must reside simultaneously in the causal subalgebra (analytic in the disk) and the anti-causal subalgebra.
Their intersection contains only constants; the strict minimum-phase condition uniquely fixes this residual zero-mode phase freedom.
This explicitly and uniquely constructs the holonomic minimum-phase section $\sigma_{mp}$ utilized for parallel transport in the main text.

\section{Bundle Topology and Connection Form Derivations}
\label{app:bundle_topology}

This appendix provides the explicit geometric derivations of the causal transport generator and the formal proof of the Flatness Theorem, which underpins the path-independence of the adaptive whitening update law.

\subsection{The Maurer-Cartan Form and the Causal Generator}
In the principal whitening bundle $\mathcal{W}$, a connection smoothly separates the tangent space at any filter into a vertical subspace (pure phase gauge transformations) and a horizontal subspace (physical changes in noise magnitude) \cite{kobayashi1963foundations, bleecker2013gauge}.
We define the minimum-phase connection by declaring the global minimum-phase section $\sigma_{mp}$ to be everywhere horizontal.
This geometric constraint forces the local gauge potential, the pullback of the connection form to the base manifold, to vanish identically: $\mathcal{A}_{mp} \equiv \sigma_{mp}^* \omega_{mp} = 0$.
Extending this structure to the full bundle via equivariance, the horizontal subspace at any arbitrary filter $W = W_{mp} \cdot U$ is generated by the pushforward of the section's geometry under right translation: $H_W = (dR_U)_* H_{W_{mp}}$ \cite{nakaharaGeometryTopologyPhysics2005, sharpe1997differential}.

Under a general change of frame $W \mapsto \sigma_{mp} \cdot U$, the connection transforms affinely via the adjoint action:
\begin{equation}
    \omega_{mp} = \text{Ad}_{U^{-1}}(\mathcal{A}_{mp}) + U^{-1} dU.
\end{equation}
Substituting the vanishing potential $\mathcal{A}_{mp} = 0$, we recover the explicit Maurer-Cartan form on the total space: $\omega_{mp} = U^{-1} dU$.

To determine how a filter must dynamically evolve to remain horizontal as the base noise state drifts along a trajectory $S(t)$, we examine the kinematics of the global section.
The transport equation requires the covariant derivative of the physical filter to vanish ($\nabla_{\dot{\gamma}} W = 0$).
Because the section is purely horizontal, its evolution in the full causal algebra $\mathcal{A}_+$ is governed entirely by its Darboux derivative, defining the instantaneous algebraic generator of the trajectory: $\Omega(t) \equiv \omega_{mp}(\dot{\gamma}) = \dot{W}_{mp}W_{mp}^{-1}$. 

For the minimum-phase filter, we defined $W_{mp} = \exp\left(\mathcal{P}_+\left[-\log S\right]\right)$.
Because the structure group $\mathcal{U}$ is abelian, the loop algebra is commutative, meaning the derivative of the exponential map simplifies exactly as it does for scalar functions: $\frac{d}{dt} e^{X(t)} = \dot{X}(t) e^{X(t)}$.
Furthermore, because the causal projection $\mathcal{P}_+$ acts strictly on the frequency domain, it commutes with the macroscopic time derivative.
Taking the time derivative of $W_{mp}$ and right-multiplying by its inverse yields the explicit causal generator.
The whitening condition dictates that the real part of the generator (the scaling flow) exactly compensates for the logarithmic noise drift.
To maintain horizontality, the generator must reside strictly within the causal subalgebra $\mathfrak{g}_{mp}$.
By the Titchmarsh theorem, this causality constraint requires the real and imaginary parts of the generator to form a Kramers-Kronig (Hilbert transform) pair.
Thus, the unique causal generator is obtained by applying the strict idempotent Riesz projection directly to the real scaling flow:
\begin{equation}
    \Omega(t) = \mathcal{P}_+\left[- \frac{\partial}{\partial t} \log S(t) \right].
\end{equation}
This establishes that the minimum-latency transport of a filter is generated purely by the causal projection of the instantaneous logarithmic drift of the power spectrum.

\subsection{Proof of the Flatness Theorem}
To determine if the parallel transport generated by $\Omega(t)$ is path-independent, we must evaluate the holonomy of the connection.
The holonomy is governed by the curvature 2-form $\Theta$, which is related to the connection form $\omega_{mp}$ via the Cartan structure equation \citep{kobayashi1963foundations, bleecker2013gauge}:
\begin{equation}
    \Theta = d\omega_{mp} + \frac{1}{2} [\omega_{mp}, \omega_{mp}].
\end{equation}
We prove that both the differential term and the algebraic commutator term vanish identically. 

First, because the structure group $\mathcal{U}$ consists of diagonal unitary operators, the corresponding loop algebra is abelian.
The Lie bracket of any two elements is therefore trivial, which immediately implies that the commutator term vanishes: $[\omega_{mp}, \omega_{mp}] \equiv 0$.

Second, we evaluate the exterior derivative of the connection form on the fibers, $d\omega_{mp} = d(U^{-1}dU)$.
Applying the product rule for exterior derivatives yields $d(U^{-1}) \wedge dU + U^{-1} d^2 U$.
By the fundamental identity of exterior calculus, $d^2 = 0$.
Utilizing the inverse differential identity $d(U^{-1}U) = 0$, we have $d(U^{-1}) = -U^{-1} dU U^{-1}$.
Because the structure group is abelian, this reduces to $-U^{-2} dU$.
Substituting this back yields $-U^{-2} dU \wedge dU$.
Because $dU$ takes values in an abelian Lie algebra, its wedge product with itself vanishes identically ($dU \wedge dU = 0$), proving that the connection form is closed: $d\omega_{mp} \equiv 0$.

The simultaneous vanishing of the Lie bracket and the exterior derivative guarantees that the curvature tensor is identically zero ($\Theta \equiv 0$), proving the flatness theorem. 

This algebraic result is deeply mirrored by the geometric topology of the bundle.
By defining our horizontal distribution as the tangent space to a global integral submanifold (the minimum-phase section $\sigma_{mp}$), the distribution is naturally involutive.
By the Frobenius theorem \citep{spivak1999comprehensive, warner1983foundations, lee2013smooth}, an involutive horizontal distribution is the necessary and sufficient geometric condition for a strictly flat connection.
Because the base manifold $\mathcal{M}$ is simply connected (as proven in Appendix \ref{app:functional_analysis}), this local flatness promotes globally to strict path-independence, ensuring the optimal filter correction $\mathcal{K}$ is entirely free of geometric hysteresis.

\section{Covariant signal transport and SNR conservation}
\label{app:snr_conservation}

In this appendix, we explicitly connect the minimum-phase connection on the principal whitening bundle $\mathcal{W}$ to the conservation of the optimal SNR via an induced covariant derivative on the physical signal space.

Let the physical signal space for a given PSD $S \in \mathcal{M}$ be the Hilbert space $\mathcal{H}_S$ equipped with the noise-weighted inner product $\langle h_1, h_2 \rangle_S = \int \tilde{h}_1^* S^{-1} \tilde{h}_2 \, df$.
A whitening filter $W \in \mathcal{W}_S$ acts as a local frame, or \textit{vielbein} (a geometric frame field that diagonalizes the noise metric), mapping the curved fiber $\mathcal{H}_S$ isometrically to a fixed flat Hilbert space $\mathcal{H}_{flat}$ via the left action $\psi = W \tilde{h}$, where $\psi$ is the whitened template. 

To define parallel transport for a signal $h(t)$ as the background noise $S(t)$ evolves along a curve $\gamma(t)$, we demand its representation in the flat frame remains constant: $\frac{d}{dt}(W \tilde{h}) = \dot{W}\tilde{h} + W\dot{\tilde{h}} = 0$.
Solving for the physical signal's evolution yields $\dot{\tilde{h}} = -W^{-1}\dot{W}\tilde{h}$.
Because the structure group is commutative, the connection 1-form evaluated along the trajectory is exactly the generator $\Omega = \dot{W}W^{-1} = W^{-1}\dot{W}$. 
Thus, we identify the induced covariant derivative on the signal bundle:
\begin{equation}
    \nabla_{\dot{\gamma}} h \equiv \dot{\tilde{h}} + \Omega \tilde{h} = 0.
\end{equation}

To guarantee that this covariant derivative preserves the metric geometry, it must be strictly metric-compatible with the time-varying inner product.
Because the noise metric factors as $S^{-1} = W^* W$, the commutativity of the scalar fields ensures the connection exactly absorbs the derivative of the metric:
\begin{equation}
\begin{split}
    \Omega^* S^{-1} + S^{-1} \Omega &= (\dot{W}^* (W^*)^{-1}) W^* W + W^* W (W^{-1} \dot{W}) \\
    &= \dot{W}^* W + W^* \dot{W} = \frac{d}{dt}(S^{-1}).
\end{split}
\end{equation}
With metric compatibility established, the optimal SNR of a signal $h(t)$, given by its scalar norm $\rho^2 = \langle h, h \rangle_S$, safely obeys the covariant Leibniz rule along the trajectory:
\begin{equation}
    \frac{d}{dt} \rho^2 = \langle \nabla_{\dot{\gamma}} h, h \rangle_S + \langle h, \nabla_{\dot{\gamma}} h \rangle_S = 0,
\end{equation}
demonstrating that parallel transport strictly conserves the theoretical SNR magnitude. 

\textit{Remark on the necessity of minimum phase:}
It is important to note that \textit{any} connection (including those accumulating arbitrary phase) preserves this scalar norm, since the unitary phase action cancels in the inner product ($U^* U = \mathbb{I}$).
However, an arbitrary connection introduces non-linear phase accumulation (dispersion), which smears the time-domain wave packet and destroys the peak of the physical matched-filter cross-correlation.
Even a linear phase connection, while preserving the wave shape, introduces a spurious time-of-arrival shift.
The minimum-phase connection ($\Omega \in \mathfrak{g}_{mp}$) is geometrically unique in its ability to simultaneously conserve the exact SNR norm and strictly anchor the astrophysical time-of-arrival by preventing any accumulation of acausal geometric phase.

\section{Operational representations of the update law}
\label{app:operational_reps}

The holonomic update law derived in the main text yields a path-independent correction kernel $\mathcal{K}$, which maps the optimal filter from a static reference state to the dynamic live state. 
In the context of a real-time signal-processing pipeline, this exact geometric correction can be deployed in three distinct but mathematically isomorphic ways. 
These representations correspond to how the kernel is distributed within the matched-filter inner product.
Because the correction acts on templates that have already been mapped to the static reference frame, this occurs strictly within the flat Hilbert space: $\rho = \langle d, \psi \rangle_{flat}$.
This choice of distribution presents markedly different computational scaling behaviors for large template banks.

The first approach is the \textit{forward action} (or pushforward).
Here, the transport operator updates the reference templates to match the live noise state: $\psi_{live} = \mathcal{K} \psi_{ref}$.
This active transformation ensures the physical waveforms exactly match the instantaneous whitening frame.
However, because this operation must be applied to every template individually, its computational cost scales linearly with the size of the template bank, $\mathcal{O}(N_{templates})$, making it computationally prohibitive for massive banks.

The second approach is the \textit{adjoint action} (or pullback).
By the fundamental definition of the adjoint in the flat Hilbert space, the action of the kernel on the template is equivalent to the action of its adjoint on the incoming data stream: $\langle d, \mathcal{K} \psi_{ref} \rangle_{flat} = \langle \mathcal{K}^\dagger d, \psi_{ref} \rangle_{flat}$.
This is the most computationally efficient architecture: it allows a massive, static bank of reference templates to remain entirely unmodified in memory while a single, dynamically updated kernel pre-conditions the live data stream.Because there is only one data stream, this correction scales as $\mathcal{O}(1)$ with respect to the template bank. 
The primary engineering drawback to this approach is causality. 
Because the drift correction kernel $\mathcal{K}$ is constructed via the causal projection $\mathcal{P}_+$, its adjoint $\mathcal{K}^\dagger$ is strictly \textit{anti-causal}. 
In a real-time control system, applying an anti-causal filter requires future knowledge of the data stream, necessitating a finite look-ahead buffer and injecting a small, fixed latency into the pipeline. 
Crucially, this latency penalty is bounded by the effective memory length (relaxation time) $\tau_{\mathcal{K}}$ of the geometric correction kernel, ensuring the system remains well within the requirements for low-latency astrophysical alerting.

The third approach is the \textit{post-processing action} (or exterior convolution).
Because the matched-filter statistic is evaluated as a time-domain cross-correlation, the associativity of linear filtering allows the kernel to be factored entirely outside the inner product. 
The pipeline can compute the naive detection statistic using the uncorrected reference template, yielding the time series $\rho_{naive}(\tau) = (d * \psi_{ref})(\tau)$, and then apply the causal geometric correction directly to this output: $\rho_{opt}(\tau) = (k * \rho_{naive})(\tau)$, where $k(\tau)$ is the time-domain impulse response of the operator $\mathcal{K}$. 
Like the forward action, this approach scales as $\mathcal{O}(N_{templates})$ because it must be applied to every output SNR stream.

Because these three operational representations are algebraically equivalent evaluations of the exact same holonomy, they all rigorously guarantee the recovery of the optimal SNR.
However, for modern gravitational-wave searches employing massive template banks, the template-agnostic adjoint action (Method 2) emerges as the strictly dominant architectural choice, provided the anti-causal buffer delay falls within acceptable low-latency alerting tolerances.

\ack{
The authors acknowledge support from the National Science Foundation under awards OAC-2103662, PHY-2308881, PHY-2011865, OAC-2201445, OAC-2018299, PHY-0757058, PHY-0823459, PHY-2207728, PHY-2513124, PHY-2110594, and PHY-2513358.
C.H. acknowledges generous support from the Pennsylvania State University Eberly College of Science, the Department of Physics, the Institute for Gravitation and the Cosmos (IGC), and the Institute for Computational and Data Sciences.

We extend distinct gratitude to the Institute for Gravitation and the Cosmos for sponsoring the ``Mathematical Aspects of Physics'' Seminar series (MAP Chat).
The geometric framework presented in this work, specifically the application of fiber bundles to signal processing, was developed directly from the interdisciplinary discussions and collaborative learning fostered within that forum.
}

\bibliographystyle{iopart-num}
\bibliography{references}

\end{document}